\documentclass{article}
\usepackage{fullpage}

\newcommand{\pagenumbaa}{1}
\usepackage{graphicx}
\usepackage{amssymb}
\usepackage{amstext}
\usepackage{epsfig}

\usepackage{latexsym}
\usepackage{amsmath}
\usepackage{bm} 
\usepackage{bbm}

\usepackage{amsthm}
\usepackage{amsmath}
\usepackage{amsfonts}
\usepackage{amssymb,amstext}
\usepackage{appendix}
\usepackage{enumitem}

\theoremstyle{plain}
\newtheorem{theorem}{Theorem}
\newtheorem{lemma}[theorem]{Lemma}
\newtheorem{corollary}[theorem]{Corollary}

\theoremstyle{definition}
\newtheorem{definition}[theorem]{Definition}

\usepackage{epsfig}
\usepackage{algpseudocode}
\usepackage{amscd}
\usepackage{algorithm}
\usepackage{tikz}
\usetikzlibrary{arrows,decorations.markings,shapes.misc}
\tikzset{cross/.style={cross out, draw=black, minimum size=2*(#1-\pgflinewidth), inner sep=0pt, outer sep=0pt},
cross/.default={1pt}}

\newcommand{\beq}{\begin{equation}}
\newcommand{\eeq}{\end{equation}}

\newcommand{\beqa}{\begin{eqnarray}}
\newcommand{\eeqa}{\end{eqnarray}}

\newcommand{\bal}{\begin{align}}
\newcommand{\eal}{\end{align}}

\newcommand{\bsp}{\begin{equation}\begin{split}}
\newcommand{\esp}{\end{split}\end{equation}}

\newcommand{\bit}{\begin{itemize}}
\newcommand{\eit}{\end{itemize}}

\newcommand{\ben}{\begin{enumerate}}
\newcommand{\een}{\end{enumerate}}

\newcommand{\nn}{\nonumber}

\newcommand{\AR}{\mathbb{R}}

\newcommand*{\ket}[1]{| #1 \rangle}
\newcommand*{\bra}[1]{\langle #1 |}

\newcommand{\ketbra}[1]{| #1 \rangle \langle #1 |}

\newcommand{\tr}{\mathrm{tr}}
\newcommand{\rank}{\mathrm{rank}}

\newcommand{\braket}[2]{\langle #1 | #2 \rangle}

\newcommand{\Herm}{\mathrm{Herm}}
\newcommand{\data}{p}

\begin{document}


\title{Learning optimal quantum models is NP-hard}


\author{Cyril J.~Stark\footnote{Center for Theoretical Physics, Massachusetts Institute of Technology, 77 Massachusetts Avenue, Cambridge MA 02139-4307, USA}}

\date{\today}

\maketitle


\begin{abstract}

Physical modeling closes the gap between perception in terms of measurements and abstraction in terms of theoretical models. Physical modeling is a major objective in physics and is generally regarded as a creative process. How good are computers at solving this task? This question is both of philosophical and practical interest because a positive answer would allow an artificial intelligence to understand the physical world. Quantum mechanics is the most fundamental physical theory and there is a deep belief that nature follows the rules of quantum mechanics. Hence, we raise the question whether computers are able to learn optimal quantum models from measured data. Here we show that in the absence of physical heuristics, the inference of optimal quantum models cannot be computed efficiently (unless \emph{P} $=$ \emph{NP}). This result illuminates rigorous limits to the extent to which computers can be used to further our understanding of nature.

\end{abstract}

\setcounter{page}{\pagenumbaa}
\thispagestyle{plain}

\section{Introduction}\label{sect:intro}

A characterization of a physical experiment is always two-fold. On the one hand, we have a description 
\[
	\mathcal{S} = \bigl( \text{description of the state} \bigr)
\]
of the state of the physical system. For instance, $\mathcal{S}$ can contain a few paragraphs of text with detailed instructions for preparing that state experimentally in the lab, or for finding it in nature.

The second part of the characterization of an experiment is the description of the measurement that is performed. As for the state, the measurement may be described in terms of a short text,
\[
	\mathcal{M} = \bigl( \text{description of the measurement} \bigr).
\]
$\mathcal{M}$ may be a complete manual for constructing the measurement device we use.

Both $\mathcal{S}$ and $\mathcal{M}$ can specify \emph{temporal} and \emph{spatial} information, e.g., the desired state is the state resulting from a particular initial state after letting it evolve for 1$\mu$s. Every experimental paper must provide $\mathcal{S}$ and $\mathcal{M}$.

Performing the measurement $\mathcal{M}$ results in a \emph{measurement outcome}. We denote by $Z$ the number of different measurement outcomes. Each of the outcomes may again be characterized in terms of a few paragraphs of text
\[
	\mathcal{O}_z = \bigl( \text{description of $z$th measurement outcome} \bigr)
\]
for all $z \in [Z] = \{ 1,...,Z \}$. Here we assume that the description $\mathcal{O}_z$ also specifies $\mathcal{M}$, i.e., it both fully specifies the measurement device and the way it signals `\emph{outcome $z$ has been measured}' to the observer.

Oftentimes we do not only consider a single state $\mathcal{S}$ and a single measurement $\bigl( \mathcal{O}_z \bigr)_{z \in [Z]}$ but $X$ states $\bigl( \mathcal{S}_x \bigr)_{x \in [X]}$ and $Y$ measurements $\bigl( \mathcal{O}_{yz} \bigr)_{z \in [Z]}$ ($y \in [Y]$). For instance, we could be interested in measuring the spin of an electron in different directions and at different times. Repeatedly measuring the state $\mathcal{S}_x$ with the measurement $\mathcal{M}_y$ we are able to collect \emph{empirical frequency distributions} $(f_{xyz})_z$ for that particular sequence of measurements. I.e., $f_{xyz} = \sharp \{ z|xy \} / N_{xy}$ where $N_{xy}$ denotes the number of times we decide to measure $\mathcal{S}_x$ with $\mathcal{M}_y$ and where $\sharp \{ z|xy \}$ denotes the number of times we measure outcome $\mathcal{O}_{yz}$ during those runs of the experiment.

To describe the experiment quantum mechanically we need to translate the verbose descriptions $\mathcal{S}_x$ and $\mathcal{O}_{yz}$ into quantum states $\rho_x$ and measurement operators $E_{yz}$. This corresponds to the task of modeling. The assignment of matrices to $\mathcal{S}_x$ and $\mathcal{O}_{yz}$ must be such that the quantum mechanical predictions are compatible with the previously measured data $f_{xyz}$. By Born's rule, $\tr(\rho_x E_{yz})$ is the probability for measuring outcome $z$ if we measure state $\mathcal{S}_x$ with the measurement $\mathcal{M}_y$. Hence, demanding compatibility between the theoretical picture $\rho_x, E_{yz}$ on the one hand and the experimental description $\mathcal{S}_x, \mathcal{O}_{yz}$ on the other hand amounts to searching states and measurements satisfying $\tr( \rho_x E_{yz} ) \approx f_{xyz}$ for all $(x,y,z) \in \Omega$. Here, $\Omega \subseteq [X] \times [Y] \times [Z]$ marks the particular combinations $(x,y,z)$ that we have measured experimentally. Combinations in the complement $(x,y,z) \in \Omega^c$ are unknown. A common pitfall to avoid is \emph{overfitting}, that is, finding an excessively complicated model that perfectly fits the data but has no predictive power over future observations.  To avoid overfitting we need to search for the lowest-dimensional model satisfying $\tr( \rho_x E_{yz} ) \approx f_{xyz}$. Note that if we placed no restriction on the dimension then we could fit every dataset exactly with a \emph{finite-dimensional} quantum model that does not allow for the prediction of future measurement outcomes. For instance, we could fit the measured data with an $X$-dimensional model where $\rho_x = \ketbra{x}$ and $E_{yz} = \sum_{x=1}^X f_{xyz} \ketbra{x}$. Indeed, $\tr(\rho_x E_{yz}) = f_{xyz}$. In contrast, if a subsystem structure (e.g., two independent parties Alice and Bob) is imposed then there are circumstances where datasets cannot be modeled by finite-dimensional quantum models~\cite{ramanathan2014no,sikora2015minimum}. 

In the remainder we are going to assume that the empirical frequencies $f_{xyz}$ are equal to the probabilities $p_{xyz}$ for measuring outcome $\mathcal{O}_{yz}$ given that we prepared $\mathcal{S}_x$ and measured $\mathcal{M}_y$. This condition is met if we can measure states $\mathcal{S}_x$ with measurements $\mathcal{M}_y$ an unbounded number of times ($N_{xy} \rightarrow \infty$). We will see that even in this noiseless setting where we want to solve
\beq\label{fewk435hjfewf0}
\begin{split}
	\text{minimize} \ & \ d \\  \text{such that} & \ \exists \text{ $d$-dimensional states and measurements}\\  &\text{ satisfying $\data_{xyz}  = \tr(\rho_x E_{yz})$ $\forall (x,y,z) \in \Omega$},
\end{split}
\eeq
inference is \emph{NP}-hard. We call problem~\eqref{fewk435hjfewf0} \emph{MinDim}; it describes the task of \emph{learning effective quantum models} from experimental data. Our result that \emph{MinDim} is \emph{NP}-hard implies that computers are not capable of computing optimal quantum models describing general experimental observations (unless $P = NP$).

\emph{NP}-hardness is a term from \emph{computational complexity theory} which aims at classifying problems according to their complexity. The relevant complexity measure depends on the particular application. Here we focus on \emph{time complexity} which measures the time it takes to solve a problem on a computer (deterministic Turing machine). A particularly important family of problems are \emph{decision problems}. These are problems whose solution is either \emph{yes} or \emph{no}. 3-coloring of graphs is a famous example. In 3-coloring (\emph{3col}) we are given a graph with vertices specified by a vertex set $V$ and with edges specified by an edge set $E$. Our task is to decide whether or not it is possible to assign colours \emph{red}, \emph{green} or \emph{blue} to  vertices $v \in V$ in such a way that vertices $v,v'$ are colored differently whenever the edge $(v,v')$ with endpoints $v,v'$ is an element of $E$. In this example, the specification of $V$ and $E$ forms the \emph{problem instance} and the criterion for the solution \emph{yes} (i.e., `\emph{yes, this graph is 3-colorable}') is the so called \emph{acceptance condition}. A decision problem is specified by an acceptance condition and by a set of problem instances. 

The complexity classes \emph{P} and \emph{NP} have been introduced to classify problems according to their complexity. The complexity class \emph{P} is the set of all decision problems whose complexity is a polynomial in the size of the problem instances (e.g., the number of vertices in case of \emph{3col}). The class \emph{NP} is the set of problems with the following property. Every \emph{yes}-instance admits a proof that can be checked in polynomial time. For example in case of \emph{3col}, we can prove that a graph is 3-colorable by providing an explicit 3-coloring of that graph; the correctness of that coloring can be verified by checking that for all $(v,v') \in E$, the vertices $v$ and $v'$ are colored differently. 

Intuitively, a problem $A$ is clearly harder to solve than a problem $B$ if any polynomial-time algorithm for $A$ that can be used to solve $B$ in polynomial time (we might use the algorithm for $A$ as a subroutine in another algorithm to solve $B$). This intuition is rigorously captured in the notion of reductions. We say that problem $B$ is \emph{reducible} to $A$ if there exists an algorithm $\mathcal{R}$ (polynomial-time) that maps problem instances $i$ for $B$ to problem instances $\mathcal{R}(i)$ for $A$ in such a way that 
\[
	\text{$i$ `yes' for $B$} \Leftrightarrow \text{$\mathcal{R}(i)$ `yes' for $A$.}
\]
Therefore, if there exists a polynomial-time algorithm to solve $A$ then this algorithm induces via $\mathcal{R}$ a polynomial time algorithm to solve $B$. A problem $A$ is \emph{NP-hard} if all problems $C \in \text{\emph{NP}}$ are reducible to $A$. For example, \emph{3col} is \emph{NP}-hard~\cite{garey1976some}.

A natural decision version of \emph{MinDim} is the problem \emph{Dim-$d$}.\\

{\emph{Dim-$d$}}. \emph{Instance}: $X,Y,Z \in \mathbb{N}$, $\Omega \subseteq [X] \times [Y] \times [Z]$ and scalars $\bigl( p_{xyz} \bigr)_{x,y,z \in \Omega}$. \emph{Acceptance condition}: there exist $d$-dimensional states $\rho_x$ and measurements $( E_{yz})_{z\in [Z]}$ such that $p_{x;yz} = \tr( \rho_x  E_{yz})$ for all $(x,y,z) \in \Omega$.\\ 

We note that \emph{Dim-$d$} outputs \emph{yes} if and only if the optimal solution $d_{\mathrm{MinDim}}$ of \emph{MinDim} satisfies $d_{\mathrm{MinDim}} \leq d$. Hence, \emph{MinDim} is NP-hard if \emph{Dim-3} is \emph{NP}-hard. In this work, we prove the latter by reduction to \emph{3col}. Thus, we are arriving at our main result, Theorem~\ref{Thm:main.theorem.for.complex.QM}.

\begin{theorem}\label{Thm:main.theorem.for.complex.QM}
	MinDim is NP-hard.
\end{theorem}

Every experiment can be described in terms of $(\mathcal{S}_x)_x$ and $(\mathcal{O})_{yz}$. Therefore, problem~\eqref{fewk435hjfewf0} does not make any assumptions about the underlying quantum model. Often, however, we do have accept some side information about the physical system we wish to analyze. A common one is the side information that we measure a global state with local measurements~\cite{popescu1994quantum,navascues2008convergent,navascues2015non,colbeck2012free}. In this setting we want to solve the modification
\beq\label{fewk435hjssqfewf}
\begin{split}
	\text{minimize} \ & \ d \\  \text{such that} & \ \exists \text{ a $d^2$-dimensional state $\rho$ and $d$-dimensional}\\ 
				&\text{ measurements $(E_{yz})_z$ and $(F_{yz})_z$ satisfying}\\  
				&\text{ $\data_{yzy'z'}  = \tr(\rho E_{yz} \otimes F_{y'z'})$ $\forall (yzy'z') \in \Omega$}
\end{split}
\eeq
(for some $\Omega \subseteq [Y] \times [Z] \times [Y'] \times [Z']$) of \emph{MinDim}. We are referring to problem~\eqref{fewk435hjssqfewf} in terms of \emph{MinDim$^{(AB)}$}; the extra label $(AB)$ references two parties usually called Alice and Bob. Here we prove \emph{NP}-hardness of \emph{MinDim$^{(AB)}$} by showing that the natural decision problem \emph{Dim-3$^{(AB)}$} (see section~\ref{sect:problem.formulation}) of \emph{MinDim$^{(AB)}$} is \emph{NP}-hard.

\begin{theorem}\label{Thm:main.theorem.for.QM.bipartite.setting}
	MinDim$^{(AB)}$ is NP-hard.
\end{theorem}

When does a physical theory qualify to be a good physical theory? Answers provided are sometimes vague. However, there is a consensus that \emph{predictive power} is a necessary criterion a good physical theory needs to satisfy. This criterion is satisfied if models drawn from that theory (e.g., quantum theory) allow for the prediction of future measurement outcomes, i.e., estimates of probabilities $p_{xyz}$ associated to pairings $(\mathcal{S}_x, \mathcal{O}_{yz})$ that have not been measured yet (i.e., $(x,y,z) \not\in \Omega$ in problem \emph{MinDim}). Therefore, considering Theorem~\ref{Thm:main.theorem.for.complex.QM} in the scenario where all probabilities $p_{xyz}$ were measured beforehand (i.e., $\Omega = [X] \times [Y] \times [Z]$) would not be very sensible because there would not be anything left to predict. A proof of hardness in that setting would, however, be of more interest in mathematical optimization where people study the optimal runtime of semidefinite program formulations of linear optimization problems~\cite{Fiorini2012,ParriloThomas2013,fawzi2014positive,ji2013binary,burgdorf2015closure}. 

Surprisingly, problem \emph{MinDim} has not yet been studied extensively~\cite{stark2014self,rosset2012imperfect,navascues2015non,stark2014compressibility,gallego2010device,harrigan2007representing,wehner2008lower,stark2012global,jackson2015detecting}. Related to \emph{MinDim} is the problem of estimating quantum processes in a way that is robust to prepare and measure errors~\cite{merkel2012self,blume2013robust,kimmel2014robust,kimmel2015robust,Dugas2014characterizing,monras2014quantum}. 

We begin by summarizing the notation (section~\ref{Sect:notation}) and the problem formulation (section~\ref{sect:problem.formulation}). We provide a proof sketch in section~\ref{sect:proof.sketch}. We prove Theorem~\ref{Thm:main.theorem.for.complex.QM} in section~\ref{sect:proof.st.m.thm} by providing separate proofs for both the complex (section~\ref{Sect:result.about.complex.case}) and real (section~\ref{Sect:result.about.real.case}) formulation of Theorem~\ref{Thm:main.theorem.for.complex.QM}. In section~\ref{sect:proof.of.bipartite.setting}, we build on top of the proof of Theorem~\ref{Thm:main.theorem.for.complex.QM} to prove Theorem~\ref{Thm:main.theorem.for.QM.bipartite.setting}. We conclude the paper in section~\ref{sect:conclusions}.

\section{Notation}\label{Sect:notation}

For any integer $n$, $[n] = \{ 1,..., n\}$. We denote graphs by $G = (V,E)$; $V$ denotes the vertex set and $E$ the edge set of the graph. For any matrix $A = (A_{ij})_{ij}$ we denote by $A^T$ its transposition and we denote by $\bar{A}$ the matrix whose $(i,j)$th entry is the complex conjugate of $A_{ij}$. Let $\Herm(\mathbb{C}^d)$ denote the space of Hermitian matrices of dimension $d$. We use $\vec{1}$ to denote the vector $(1,...,1)^T$ and $\rho \in S^+(\mathbb{C}^d)$ to denote the set of complex positive semidefinite (psd) matrices. Quantum states of dimension $d$ are specified in terms complex psd matrices with trace 1, i.e., $\rho \in S^+(\mathbb{C}^d)$ and $\tr(\rho) = 1$. A $d$-dimensional quantum description of a measurement device with $Z$ outcomes is specified in terms of psd matrices $E_1, ..., E_Z \in \Herm(\mathbb{C}^d)$ subject to the constraint $\sum_{z=1}^Z E_z = I_d$. Here, $I_d$ denotes the identity matrix on $\mathbb{C}^d$. By the so called Born's rule, the probability for obtaining outcome $z$ when measuring a state $\rho$ with the measurement $(E_z)_{z=1}^Z$ is $\tr(\rho E_z)$.

We denote by $\data_{xyz}$ the probability for measuring outcome $z$ given that we measure the state $\rho_x$ with the measurement $(E_{yz})_{z=1}^Z$. We will frequently refer to the data $(\data_{xyz})_{xyz}$ in terms of a matrix $\data \in \mathbb{R}^{X \times YZ}$,
\beq\label{def.data.D.0}
	\data =   \left( \begin{array}{ccccccc}  \data_{111} & \cdots & \data_{11Z} & \cdots &   \data_{1Y1} & \cdots & \data_{1YZ}  \\  \data_{211} & \cdots & \data_{21Z} & \cdots &   \data_{2Y1} & \cdots & \data_{2YZ} \\  \vdots &   & \vdots &   &   \vdots &   & \vdots  \\   \data_{X11} & \cdots & \data_{X1Z} & \cdots &   \data_{XY1} & \cdots & \data_{XYZ}  \end{array} \right) .
\eeq
The matrix $p$ is a flattening of $(\data_{xyz})_{xyz}$. Note that changing the row amounts to changing the state and changing the column amounts to changing the measurements outcome. The first $Z$ columns capture all the probabilities associated to the first measurement, the columns $Z+1$ to $2Z$ capture all the probabilities associated to the second measurement, etc.

\section{Problem formulation}\label{sect:problem.formulation}

Oftentimes we do not know how to describe the experimental states and measurements in terms of explicit matrices $\rho_x$ and $(E_{yz})_{z=1}^Z$. By measuring different states with different measurements we only have access to empirical distributions for obtaining outcomes $z \in [Z]$ given that we prepared state $\rho_x$ and given that we measured $(E_{yz})_{z=1}^Z$. We denote the values of the corresponding probability distributions by $\data_{xyz}$, i.e., $\mathbb{P}[z|xy] = \data_{xyz}$. By Born's rule, $\data_{xyz} = \tr(\rho_x E_{yz})$.

Hence, to find a low-dimensional quantum model for the considered experiment, we aim at solving the problem
\beq\label{fewk435hjfewf}
\begin{split}
	\min d \  \text{s.t.} & \ \exists \text{$d$-dimensional states and measurements}\\  &\text{ such that $\data_{xyz}  = \tr(\rho_x E_{yz})$ $\forall (x,y,z) \in \Omega$}.
\end{split}
\eeq
for some index set $\Omega \subset [X] \times [Y] \times [Z]$ that marks those $p_{xyz}$ that have been measured experimentally. Closely related is the problem
\beq\label{fewk435hjfewfdwdwd}
\begin{split}
	\text{minimize} \ & \ d \\  \text{such that} & \ \exists \text{ a $d^2$-dimensional state $\rho$ and $d$-dimensional}\\ 
				&\text{ measurements $(E_{yz})_z$ and $(F_{yz})_z$ satisfying}\\  
				&\text{ $\data_{yzy'z'}  = \tr(\rho E_{yz} \otimes F_{y'z'})$ $\forall (yzy'z') \in \Omega$}
\end{split}
\eeq
which appears in the study of nonlocal correlations. In the remainder we are mainly referring to the following problems:	\\

\begin{itemize}
	\item \emph{3col}. This is the following decision problem. \emph{Instance}: a graph $G = (V,E)$. \emph{Acceptance condition}: there exists a function $\mathrm{c}: V \rightarrow \{ \mathrm{r}, \mathrm{g}, \mathrm{b} \}$ such that for all $v,v' \in V$ with $( v,v' ) \in E$, $\mathrm{c}(v) \neq \mathrm{c}(v')$.
	\item \emph{Dim-$d$}. This is the following decision problem. \emph{Instance}: $X,Y,Z \in \mathbb{N}$, $\Omega \subseteq [X] \times [Y] \times [Z]$ and scalars $\bigl( p_{xyz} \bigr)_{x,y,z \in \Omega}$. \emph{Acceptance condition}: there exist $d$-dimensional states $\rho_x$ and measurements $( E_{yz})_{z\in [Z]}$ such that $p_{xyz} = \tr( \rho_x  E_{yz})$ for all $(x,y,z) \in \Omega$. 
	\item \emph{MinDim}. This is the following optimization problem. \emph{Instance}: $X,Y,Z \in \mathbb{N}$, $\Omega \subseteq [X] \times [Y] \times [Z]$ and scalars $\bigl( p_{xyz} \bigr)_{x,y,z \in \Omega}$. \emph{Objective}: see~\eqref{fewk435hjfewf}.
	\item \emph{Dim-$d^{(AB)}$}. This is the following decision problem. \emph{Instance}: $Y,Z,Y',Z' \in \mathbb{N}$, $\Omega \subseteq [Y] \times [Z] \times [Y'] \times [Z']$ and scalars $\bigl( p_{yzy'z'} \bigr)_{y,z,y'z' \in \Omega}$. \emph{Acceptance condition}: there exists a $d^2$-dimensional state $\rho$ in $\Herm(\mathbb{C}^d \otimes \mathbb{C}^d)$ and measurements $( E_{yz})_{z\in [Z]}$ and $( F_{y'z'})_{z'\in [Z']}$ such that $p_{yzy'z'} = \tr( \rho  E_{yz} \otimes F_{y'z'})$ for all $(y,z,y',z') \in \Omega$. 
	\item \emph{MinDim}$^{(AB)}$. This is the following optimization problem. \emph{Instance}: $Y,Z,Y',Z' \in \mathbb{N}$, $\Omega \subseteq [Y] \times [Z] \times [Y'] \times [Z']$ and scalars $\bigl( p_{yzy'z'} \bigr)_{y,z,y'z' \in \Omega}$. \emph{Objective}: see~\eqref{fewk435hjfewfdwdwd}.
\end{itemize}

\section{Proof sketch}\label{sect:proof.sketch}

We prove Theorem~\ref{Thm:main.theorem.for.complex.QM} by showing that \emph{Dim-3} is \emph{NP}-hard. Figure~\ref{fig:direct.reduction.to.Dim3.via.3col} sketches the strategy of our proof. We construct a sequence of reductions whose composition reduces \emph{3col} to \emph{Dim-3}. This suffices to prove the theorem because \emph{3col} is known to be \emph{NP}-hard~\cite{garey1976some}. Analogously, we prove Theorem~\ref{Thm:main.theorem.for.QM.bipartite.setting} by showing that the associated decision problem \emph{Dim-3$^{(AB)}$} is \emph{NP}-hard.

\begin{figure}[tbp]
\centering
\begin{tikzpicture}[scale=0.9]
	
	\draw[thick,rounded corners=12pt,fill=red!15] (-0.5,-2.9) rectangle (8.5,3.7);
	\node at (0.3,3.2)[]{\large \emph{NP}};
	
	\draw[thick,rounded corners=8pt,fill=blue!18] (0,0) rectangle (2,0.8);
	\node at (1,0.4)[]{\emph{3col}};
	
	\draw[thick,rounded corners=8pt,fill=blue!18] (3,0) rectangle (5,0.8);
	\node at (4,0.44)[]{\emph{rank-3}$^{\Delta}$};
	
	\draw[thick,rounded corners=8pt,fill=blue!18] (6,0) rectangle (8,0.8);
	\node at (7,0.44)[]{\emph{Dim-3}$^{\text{(pre)}}$};
	
	\draw[thick,rounded corners=8pt,fill=blue!18] (6,-1.6) rectangle (8,-2.4);
	\node at (7,-1.96)[]{\emph{Dim-3$^{(AB)}$}};
	
	\draw[thick,rounded corners=8pt,fill=blue!18] (6,2.4) rectangle (8,3.2);
	\node at (7,2.8)[]{\emph{Dim-3}};
	  
	\draw [thick,decoration={markings,mark=at position 1 with
    	{\arrow[scale=2,>=stealth]{>}}},postaction={decorate}] (2,0.4) -- (3,0.4);
	
	\node at (2.42,0.1){$\mathcal{A}_1$};
	
	\draw [thick,decoration={markings,mark=at position 1 with
    	{\arrow[scale=2,>=stealth]{>}}},postaction={decorate}] (5,0.4) -- (6,0.4);
	
	\node at (5.42,0.1){$\mathcal{A}_2$};
	
	\draw [thick,decoration={markings,mark=at position 1 with
    	{\arrow[scale=2,>=stealth]{>}}},postaction={decorate}] (7,0.8) -- (7,2.4);
	
	\node at (7.37,1.6){$\mathcal{A}_3$};
	
	\draw [thick,decoration={markings,mark=at position 1 with
    	{\arrow[scale=2,>=stealth]{>}}},postaction={decorate}] (7,0) -- (7,-1.6);
	
	\node at (7.37,-0.8){$\mathcal{A}'_3$};
	
	\node at (-0.09,1.58){$\mathcal{A}_0$};
	
	\draw[->,thick] (0.3,2.54) to [out=255,in=108] (0.4,0.83);
	\draw[->,thick] (1.3,2.6) to [out=242,in=94] (0.8,0.83);
	\draw[->,thick] (2.5,2.55) to [out=210,in=82] (1.15,0.83);
	\draw[->,thick] (0.3,-1.74) to [out=105,in=252] (0.4,-0.03);
	\draw[->,thick] (1.3,-1.8) to [out=118,in=266] (0.8,-0.03);
	\draw[->,thick] (2.5,-1.75) to [out=150,in=278] (1.15,-0.03);
	
	\draw (0.3,2.54) node[thick,cross=2.3pt,rotate=255] {};
	\draw (1.3,2.6) node[thick,cross=2.3pt,rotate=242] {};
	\draw (2.5,2.55) node[thick,cross=2.3pt,rotate=210] {};
	\draw (0.3,-1.74) node[thick,cross=2.3pt,rotate=105] {};
	\draw (1.3,-1.8) node[thick,cross=2.3pt,rotate=118] {};
	\draw (2.5,-1.75) node[thick,cross=2.3pt,rotate=150] {};

\end{tikzpicture}
\caption{Successive reduction from problems in \emph{NP} to \emph{Dim-3}.}
\label{fig:direct.reduction.to.Dim3.via.3col}
\end{figure}
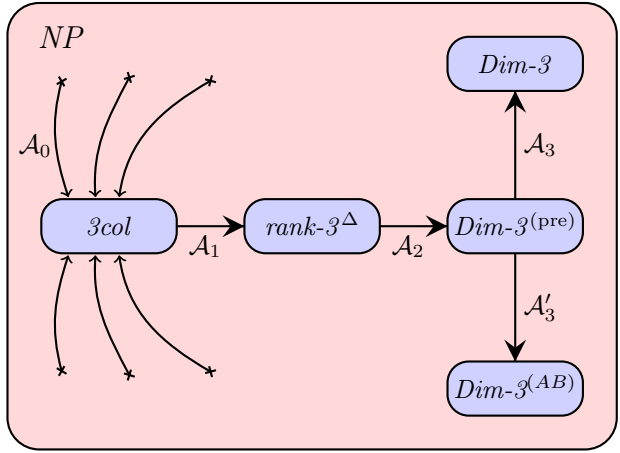

Thus, to prove Theorem~\ref{Thm:main.theorem.for.complex.QM}, we need to find a polynomial-time algorithm $\mathcal{A}$ that maps instances for \emph{3col} to instances of \emph{Dim-3} such that an instance $i$ for \emph{3col} is a \emph{yes}-instance for \emph{3col} if and only if $\mathcal{A}(i)$ is a \emph{yes}-instance for \emph{Dim-3}. As suggested by figure~\ref{fig:direct.reduction.to.Dim3.via.3col}, the reduction $\mathcal{A}$ is the composition of several partial reductions, i.e., $\mathcal{A} = \mathcal{A}_3 \circ \mathcal{A}_2 \circ \mathcal{A}_1$. Each of the parts $\mathcal{A}_1$, $\mathcal{A}_2$, $\mathcal{A}_3$ are defined in the remainder of this section. The reduction $\mathcal{A}_0$ from any problem in \emph{NP} to \emph{3col} is introduced in~\cite{garey1976some}. Consequently, reductions $\mathcal{A} \circ \mathcal{A}_0$ reduce any problem in \emph{NP} to \emph{Dim-3}. 

In section~\ref{sect:proof.st.m.thm} we provide the analysis of the algorithms $\mathcal{A}_j$ and the formal proof of Theorem~\ref{Thm:main.theorem.for.complex.QM}. Similarly, to prove Theorem~\ref{Thm:main.theorem.for.QM.bipartite.setting} we provide a reduction $\mathcal{A}' =  \mathcal{A}'_3 \circ \mathcal{A}_2 \circ \mathcal{A}_1$ from \emph{3col} to \emph{Dim-3$^{(AB)}$}. Here, the sub-reductions $\mathcal{A}_1$ and $\mathcal{A}_2$ are identical to the sub-reductions used in the proof of Theorem~\ref{Thm:main.theorem.for.complex.QM}. Only the last sub-reduction $\mathcal{A}_3$ requires modification. That modification $\mathcal{A}_3'$ and its discussion are provided in section~\ref{sect:proof.of.bipartite.setting}.

The reduction $\mathcal{A}_1$ proceeds in two steps. Firstly, for each pair of vertices $v,v'$ of $G$, it inserts subgraphs $H_{vv'}$ from~\cite{peeters1996orthogonal} (see figure~\ref{Fig:Peeters.gadget}) into $G$. We call the resulting graph $G'$. Then, in the second step, $\mathcal{A}_1$ adds triangles to each vertex of $G'$. This second operation returns a graph that we call $\Delta(G')$; see figure~\ref{Fig:example.for.triangle.decoration}. 

Using a slight modification of a Theorem from~\cite{peeters1996orthogonal} we can show that $G$ is 3-colorable if and only if there exists a Gram matrix $A$ whose rank is $\leq 3$ and whose entries are subject to linear constraints described in terms of the graph $\Delta(G')$. The decision problem \emph{rank-3}$^{\Delta}$ is defined through the question whether or not such a matrix $A$ exists.

By figure~\ref{fig:direct.reduction.to.Dim3.via.3col}, the next step in the reduction $\mathcal{A}$ is the transition to \emph{Dim-3}$^{\text{(pre)}}$. In the following definition of \emph{Dim-3}$^{\text{(pre)}}$, $y$ enumerates the vertices of the graph $G'$ and for each of these vertices, $(y1)$, $(y2)$, $(y3)$ enumerates the vertices of the triangle $\subseteq \Delta(G')$ attached to the vertex labelled by $y$.\\

\emph{Dim-3}$^{\text{(pre)}}$. \emph{Instance}: identical to \emph{rank-3}$^{\Delta}$, i.e., graphs $\Delta(G')$ where $G$ is an arbitrary graph. Hence, $\mathcal{A}_2$ is the identity function. \emph{Acceptance condition}: there exist vectors $\vec{v}_{yz} \in \mathbb{C}^3$ such that the matrix $\data$ defined by $\data_{yz,y'z'} := | \braket{v_{yz}}{v_{y'z'}} |^2$ satisfies the following: firstly, $\data_{yz,yz} = 1$ for all $(yz) \in \Delta(G')$ and secondly, $\data_{yz,y'z'} = 0$ whenever $(yz,y'z')$ is in the edge set of $\Delta(G')$. \\

The claim that $\mathcal{A}_2$ is a valid reduction is almost a direct consequence of $1^2 = 1$ and $0^2 = 0$. The last step of the reduction $\mathcal{A}$ is the transition from \emph{Dim-3}$^{\text{(pre)}}$ to \emph{Dim-3}. Instances of \emph{Dim-3} are tuples $\bigl( p_{xyz} \bigr)_{x,y,z \in \Omega}$. Hence, $\mathcal{A}_3$ must map graphs $\Delta(G')$ to such tuples. To define the action of $\mathcal{A}_3$ we proceed as in the definition of \emph{Dim-3}$^{\text{(pre)}}$ by enumerating the vertices of $G'$ by $y$. The vertices of the triangle attached to $y$ are labeled by $(yz)$ with $z = 1,2,3$. This allows us to define an index set $\Omega'$ as follows. We start by setting $\Omega'$ equal to the empty set. Then, for each edge $(yz,y'z')$ in the edge set of $\Delta(G')$, we add $\bigl( 3 (y-1) + z, y',z' \bigr)$ to $\Omega'$. Next we demand that the probabilities $\data$ from the definition of \emph{Dim-3} satisfy
\[
	\data_{x,yz} = \left\{  \begin{array}{ll}
  				1,	& \text{ if $x = 3 (y-1) + z$,}  \\
  				0,	& \text{ if $(x,y,z) \in \Omega'$}.   
\end{array}
 \right.
\]
These constraints constitute the input $\mathcal{A}_3\bigl( \Delta(G) \bigr)$ to \emph{Dim-3}. The proof of Theorem~\ref{Thm:main.theorem.for.QM.bipartite.setting} proceeds along the same lines. We only need to modify the reduction $\mathcal{A}_3$ so that the output of $\mathcal{A}_3'$ forms a valid input to \emph{Dim-3$^{(AB)}$}.

\begin{figure}[tbp]
\centering
\begin{tikzpicture}[scale=1.1]
	\draw[thick] (0,1) -- (1.5,0);
	\draw[thick] (1.5,0) -- (3,0.3);
	\draw[thick] (3,0.3) -- (0,1);
	\draw[thick] (3,0.3) -- (2.7,1.2);
	\draw[thick] (2.7,1.2) -- (1.4,1.8);
	\draw[thick] (1.4,1.8) -- (0,1);		
	
	\node at (4.4,0.8){$\longmapsto$};	
	
	\draw[thick] (6,1) -- (7.5,0);
	\draw[thick] (7.5,0) -- (9,0.3);
	\draw[thick] (9,0.3) -- (6,1);
	\draw[thick] (9,0.3) -- (8.7,1.2);
	\draw[thick] (8.7,1.2) -- (7.4,1.8);
	\draw[thick] (7.4,1.8) -- (6,1);
	
	\draw[thick] (6,1) -- (5.8,0.85);
	\draw[thick] (5.8,0.85) -- (5.8,1.15);
	\draw[thick] (5.8,1.15) -- (6,1);
	
	\draw[thick] (7.5,0) -- (7.65,-0.15);
	\draw[thick] (7.65,-0.15) -- (7.35,-0.15);
	\draw[thick] (7.35,-0.15) -- (7.5,0);
	
	\draw[thick] (9,0.3) -- (9.14,0.14);
	\draw[thick] (9.14,0.14) -- (9.22,0.43);
	\draw[thick] (9.22,0.43) -- (9,0.3);
	
	\draw[thick] (8.7,1.2) -- (9,1.35);
	\draw[thick] (9,1.35) -- (8.65,1.5);
	\draw[thick] (8.65,1.5) -- (8.7,1.2);
	
	\draw[thick] (7.4,1.8) -- (7.6,2);
	\draw[thick] (7.6,2) -- (7.3,2);
	\draw[thick] (7.3,2) -- (7.4,1.8);
\end{tikzpicture}
\caption{Illustration of triangle decoration $\Delta$ of a graph. The left hand side displays an arbitrary graph $G = (V,E)$. The right hand side shows the triangle decoration $\Delta(G) = (\Delta(V),\Delta(E))$ of $G$.}
\label{Fig:example.for.triangle.decoration}
\end{figure}
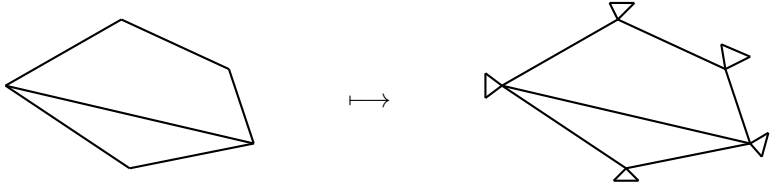

\section{Proof of Theorem~\ref{Thm:main.theorem.for.complex.QM}}\label{sect:proof.st.m.thm}

\subsection{Proof in the complex case}\label{Sect:result.about.complex.case}

Recall the reduction $\mathcal{A} = \mathcal{A}_3 \circ \mathcal{A}_2 \circ \mathcal{A}_1$ and the decision problems \emph{3col}, \emph{rank-3}$^{\Delta}$, \emph{Dim-3}$^{\text{(pre)}}$ and \emph{Dim-3} from section~\ref{sect:proof.sketch}. Here, we use the proof strategy sketched in Figure~\ref{fig:direct.reduction.to.Dim3.via.3col} to prove Theorem~\ref{Thm:main.theorem.for.complex.QM}.

\begin{definition}[see~\cite{peeters1996orthogonal}]
	From $G = (V,E)$ we construct $G' = (V',E')$ as follows. For \emph{every} pair $\{ i,j \} \subseteq V$ with $i \neq j$ we add new vertices $a_{ij},b_{ij},c_{ij},d_{ij}$ to $V$. This yields $V'$. As dictated by $H_{ij}$ from figure~\ref{Fig:Peeters.gadget}, we connect the vertices $a_{ij},b_{ij},c_{ij},d_{ij}$ to $\{ i,j \}$ and among themselves. The resulting graph is $G' = (V',E')$ with $|V'| = 4 n(n-1)/2$ and $|E'| = 9 n(n-1)/2$.
\end{definition}

\begin{figure}
\centering
\begin{tikzpicture}[scale=0.9]
	\draw[thick] (0,0) -- (0,3) ;
	\draw[thick] (0,0) -- (2,1.5);
	\draw[thick] (2,1.5) -- (0,3);
	\draw[thick] (2,1.5) -- (4,1.5);
	\draw[thick] (4,1.5) -- (6,0);
	\draw[thick] (6,0) -- (6,3);
	\draw[thick] (6,3) -- (4,1.5);
	\draw[thick] (6,3) -- (0,3);
	\draw[thick] (6,0) -- (0,0);
	\fill (0,0)  circle[radius=3pt] node[anchor=north east] {$a_{ij}$};
	\fill (0,3)  circle[radius=3pt] node[anchor=south east] {$i$};
	\fill (2,1.5)  circle[radius=3pt] node[anchor=south] {$b_{ij}$};
	\fill (4,1.5)  circle[radius=3pt] node[anchor=south] {$c_{ij}$};
	\fill (6,0)  circle[radius=3pt] node[anchor=north west] {$j$};
	\fill (6,3)  circle[radius=3pt] node[anchor=south west] {$d_{ij}$};
\end{tikzpicture}
\caption{Graph $H_{ij}$ from~\cite{peeters1996orthogonal}.}
\label{Fig:Peeters.gadget}
\end{figure}
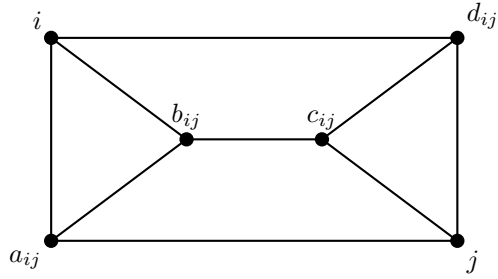

\begin{definition}[see~\cite{peeters1996orthogonal}]\label{Def:fitting.of.a.graph}
	Let $G=(V,E)$ be a graph. A matrix $A \in \mathbb{C}^{|V| \times |V|}$ is said to \emph{fit} $G$ if
	\begin{itemize}
	\item $A_{jj} = 1$ for all $j \in V$, and if
	\item	$A_{ij} = 0$ for all $(i,j) \in E$. 
	\end{itemize}
\end{definition}

\begin{theorem}[Gram matrix version of~\cite{peeters1996orthogonal}]\label{Thm:peeters.1st.formulation}
	$G$ is 3-colorable if and only if there exists a Gram matrix $A$ such that $\rank(A) \leq 3$ and such that $A$ fits $G'$.
\end{theorem}

In appendix~\ref{Sect:derivation.peeters} we provide a sketch of the arguments from~\cite{peeters1996orthogonal} to prove Theorem~\ref{Thm:peeters.1st.formulation}.

\begin{definition}
	Let $G =(V,E)$ be a graph. We call $\Delta(G) = (\Delta(V), \Delta(E))$ a \emph{triangle decoration} of $G$ if the following applies:
	\begin{itemize}
	\item $\Delta(V) = \bigcup_{v \in V} \{ v, q_v, t_v \}$ for some new vertices $\{ q_v, t_v \}_v$.
	\item For all $v, v' \in V$, $( v,v' ) \in \Delta(E)$ if and only if $( v,v' ) \in E$.
	\item	For all $v \in V$, $( v, q_v ) \in \Delta(E)$, $( v, t_v ) \in \Delta(E)$ and $( q_v, t_v ) \in \Delta(E)$. 
	\end{itemize}
\end{definition}

An example of a triangle decoration of a graph is shown in Fig.~\ref{Fig:example.for.triangle.decoration}.

\begin{lemma}[Reduction~$\mathcal{A}_1$]\label{Lemma:gjerkgjrlk}
	The following are equivalent.
	\begin{itemize}
	\item $G$ is 3-colorable.
	\item There exists a Gram matrix $A$ with $\rank(A) \leq 3$ such that $A$ fits $\Delta(G')$.
	\end{itemize}
\end{lemma}

\begin{proof}
	``$\Rightarrow$": By Theorem~\ref{Thm:peeters.1st.formulation} there exists a Gram matrix $A$ such that $\rank(A) \leq 3$ and such that $A$ fits $G'$. Let $P \in \mathbb{C}^{3 \times |V'|}$ be such that $A = \bar{P}^T P$. For each $v \in V'$ we denote by $\vec{P}_v$ the $v$-th column of $P$. Each of those column vectors has unit length and can (separately for each $v$) be completed to an orthonormal basis $\{ \vec{P}_v, \vec{P}_{q_v}, \vec{P}_{t_v} \}$ for some abstract labels $q_v$ and $t_v$. The identification of $q_v$ and $t_v$ with the vertices from $\Delta(G')$ proves the claim.
	
	``$\Leftarrow$": $V' \subset \Delta(V')$ and $E' \subset \Delta(E')$. Therefore, restricting $A$ to the submatrix of $A$ corresponding to $G'$ yields a matrix $B$ with the following properties. Firstly, $\rank(B) \leq 3$ and secondly, $B$ fits $G'$. By Theorem~\ref{Thm:peeters.1st.formulation}, this suffices to prove the claim.
\end{proof}

\begin{lemma}[Reduction~$\mathcal{A}_2$]\label{Lemma:fiehfheferf}
	The following are equivalent:
	\begin{itemize}
	\item There exists a Gram matrix $A$ with $\rank(A) \leq 3$ such that $A$ fits $\Delta(G')$.
	\item There exists $\psi_{yz} \in \mathbb{C}^3$ such that the matrix $\data$ with entries $\data_{yz;y'z'} := | \bar{\psi}_{yz}^T \psi_{y'z'} |^2$ fits $\Delta(G')$. Here, $y,y' \in [|V|]$ and $z,z' \in [3]$.
	\end{itemize}
\end{lemma}

\begin{proof}

``$\Rightarrow$": The matrix $A$ is of size $3|V'| \times 3|V'|$. Therefore, for each entry $A_{nn'}$ there exist $y,y' \in [|V'|]$ and $z,z' \in [3]$ such that $n = 3(y-1) + z$ and $n' = 3(y'-1) + z'$. This allows us to use double indices to refer to matrix elements of $A$, i.e., $A_{nn'} = A_{yz,y'z'}$. Since $A$ is a Gram matrix with rank $\leq 3$, there exists $\psi_{yz} \in \mathbb{C}^3$ such that $A_{yz,y'z'} = \bar{\psi}_{yz}^T \psi_{y'z'}$. Denote by $A' \in \AR^{3 |V'| \times 3 |V'|}$ the matrix defined by 
	\[
		(A')_{yz,y'z'} := | A_{yz,y'z'} |^2 = | \bar{\psi}_{yz}^T \psi_{y'z'} |^2.
	\] 
	By definition~\ref{Def:fitting.of.a.graph}, $A'$ fits $\Delta(G')$ because all the conditions in definition~\ref{Def:fitting.of.a.graph} are formulated in terms of entries of $A$ which are equal to $0$ and $1$. These entries are invariant under the transition $A \mapsto A'$.
	
``$\Leftarrow$": Defining $A \in \mathbb{C}^{3 |V| \times 3 |V|}$ by $A_{yz,y'z'} := \bar{\psi}_{yz}^T \psi_{y'z'}$ proves the claim.
\end{proof}

\begin{lemma}[Reduction~$\mathcal{A}_3$]\label{Lemma:kjfwenNKer}
	The following are equivalent:
	\begin{itemize}
	\item There exists $\psi_{yz} \in \mathbb{C}^3$ such that the matrix $\data$ with entries $\data_{yz;y'z'} := | \bar{\psi}_{yz}^T \psi_{y'z'} |^2$ fits $\Delta(G')$. Here, $y,y' \in [|V'|]$ and $z,z' \in [3]$.
	\item For $X = 3 |V'|$, $Y=|V'|$ and $Z=3$ there exists a 3-dimensional quantum model with the property that the matrix $\data$ defined by $\data_{x;y'z'} := \tr( \rho_{x} E_{y'z'} )$ fits $\Delta(G')$. Here, $x \in [3|V'|]$, $y' \in [|V'|]$ and $z' \in [3]$.
	\end{itemize}
\end{lemma}

\begin{proof}

``$\Rightarrow$": This is the easy direction. Setting $\rho_{3(y-1) + z} := \psi_{yz} \bar{\psi}_{yz}^T$ and $E_{yz} := \psi_{yz} \bar{\psi}_{yz}^T$ proves the claim.

``$\Leftarrow$": As in the proof of Lemma~\ref{Lemma:fiehfheferf} we replace the $x$-index of states with a $(yz)$-index so that $\rho_x = \rho_{3(y-1) + z} = \rho_{yz}$. By Cauchy-Schwarz
\beq\label{fej43jifdnnfwefew}
	\| \rho_{yz} \|_2 \| E_{yz} \|_2 \geq \tr(\rho_{yz} E_{yz}) = 1
\eeq
for all $y \in [|V'|]$ and $z \in [3]$. It follows that 
\beq\label{feij3iojwfefjfe}
	\| E_{yz} \|_2 \geq 1
\eeq
because $\| \sigma \|_2 \in [1/\sqrt{Z},1]$ for all $Z$-dimensional quantum states $\sigma$. By summation to the identity of measurements and by the self-duality of the cone of positive semidefinite matrices,
	\beq\begin{split}\nn
		Z 
		&= \| I \|^2_2 = \Bigl\| \sum_{z=1}^Z E_{yz} \|_2^2 \\
		&= \Bigl( \sum_{z=1}^Z \| E_{yz} \|_2^2 \Bigr) + \sum_{z \neq z'} \tr(E_{yz} E_{yz'}).
	\end{split}\eeq
	Therefore,
	\beq\label{fwjh453ksk}
		Z \geq \sum_{z=1}^Z \| E_{yz} \|_2^2 
	\eeq
	because $\tr(MN) \geq 0$ for any positive semidefinite matrices $M,N$. By~\eqref{fwjh453ksk}, there exists $z'$ such that $\| E_{yz'} \|_2 \leq 1$. Assume there exists $z^*$ such that $\| E_{yz^*} \|_2 > 1$. Then, by~\eqref{feij3iojwfefjfe},
	\beq\begin{split}\nn
		Z 
		&\geq \sum_{z=1}^Z \| E_{yz} \|_2^2 \geq \| E_{yz^*} \|_2^2 + (Z-1) \min_z \| E_{yz} \|^2_2 \\
		&>  Z \min_z \| E_{yz} \|^2_2 \geq Z.
	\end{split}\eeq
This is impossible and therefore, $\| E_{yz} \|_2 = 1$ for all $y \in [| V' |]$ and $z \in [3]$. By~\eqref{fej43jifdnnfwefew},
\beq
	\| \rho_{yz} \|_2 = 1,
\eeq
i.e., all states are pure and the Cauchy-Schwarz inequality~\eqref{fej43jifdnnfwefew} is satisfied with equality. This happens if and only if there exists $\kappa \in \AR$ such that $\rho_{yz} = \kappa E_{yz}$. From $\| \rho_{yz} \|_2 = \| E_{yz} \|_2$ we conclude that $\kappa \in \{ \pm 1 \}$. The possibility $\kappa = -1$ can be ruled out because both $\rho_{yz}$ and $E_{yz}$ are positive semidefinite. We conclude that for all $y \in [|V'|]$ and $z \in [3]$ there exist unit vectors $\vec{\psi}_{yz} \in \mathbb{C}^3$ such that
\beq\label{lfklek5dwdwdw469kflrk}
	| \bar{\psi}_{yz}^T \psi_{y'z'} |^2 = \data_{yz,y'z'}.
\eeq
This proves the claim because by assumption, $\data$ fits $\Delta(G')$.

\end{proof}

\begin{corollary}\label{Corollary:fjkerww3JHfe}
	The following are equivalent:
	\begin{itemize}
	\item $G$ is 3-colorable.
	\item For $X = 3 |V'|$, $Y=|V'|$ and $Z=3$ there exists a 3-dimensional quantum model with the property that the matrix $\data \in \AR^{3|V'| \times 3|V'|}$ defined by $\data_{yz;y'z'} := \tr( \rho_{yz} E_{y'z'} )$ fits $\Delta(G')$.
	\end{itemize}
\end{corollary}

\begin{proof}
	The claim is the straightforward combination of the statements of Lemma~\ref{Lemma:gjerkgjrlk}, Lemma~\ref{Lemma:fiehfheferf} and Lemma~\ref{Lemma:kjfwenNKer}.
\end{proof}

\begin{corollary}\label{Corollary:fewfjs3jkjcc}
	Dim-3 is NP-hard.
\end{corollary}

\begin{proof}
	Corollary~\ref{Corollary:fjkerww3JHfe} reduces \emph{3col} to \emph{Dim-3}. Therefore, \emph{Dim-3} is \emph{NP}-hard because \emph{3col} is \emph{NP}-complete~\cite{garey1976some}.
\end{proof}

Corollary~\ref{Corollary:fewfjs3jkjcc} is sufficient to prove \emph{NP}-hardness of \emph{MinDim} because \emph{Dim-3} can be reduced to \emph{MinDim} by checking whether or not the optimal dimension computed by \emph{MinDim} is $\leq 3$. This concludes the proof of Theorem~\ref{Thm:main.theorem.for.complex.QM}.

\subsection{Proof in the real case}\label{Sect:result.about.real.case}

We provide a separate proof for the natural real-valued formulation of Theorem~\ref{Thm:main.theorem.for.complex.QM} because this proof of Theorem~\ref{Thm:main.theorem.for.complex.QM} is more self-contained and is based on a reduction from the partition problem (instead of \emph{3col}). Hence, the goal of this section is to prove that~\eqref{fewk435hjfewf} is \emph{NP}-hard in the real case. I.e., we want to show that~\eqref{fewk435hjfewf} is \emph{NP}-hard when the quantum states and measurements are enforced to be matrices with real matrix entries. This leads to the consideration of the real-valued variants of \emph{Dim-3} and \emph{MinDim}:
\begin{itemize}
\item \emph{$\mathbb{R}$-Dim-$d$}. This is the natural real variant of the decision problem \emph{Dim-$d$}, i.e., the underlying Hilbert space is a real vector space.
\item \emph{$\mathbb{R}$-MinDim}. This is the natural real variant of the optimization problem \emph{MinDim}, i.e., the underlying Hilbert space is a real vector space.
\end{itemize}

The real formulation of Theorem~\ref{Thm:main.theorem.for.complex.QM} is a corollary of the following Theorem~\ref{Thm:main.result.vers2}.

\begin{theorem}\label{Thm:main.result.vers2}
	The decision problem $\AR$-Dim-$d$ is NP-hard.
\end{theorem}

Inspired by the proof~\cite{fawzi2014positive} of \emph{NP}-hardness of square root-rank, we are going to prove Theorm~\ref{Thm:main.result.vers2} by reducing the partition problem to $\AR$-Dim-$d$;\\

\emph{Partition problem.} This is the following decision problem. \emph{Instance}: $c_1, ..., c_Z \in \mathbb{N}$.  \emph{Acceptance condition}: accept if and only if there exist signs $s_1, ..., s_Z \in \{ \pm 1 \}$ such that $\sum_{j=1}^Z s_j c_j = 0$.\\

The reduction from the partition problem to $\AR$-Dim-$d$ suffices to prove the claim because the partition problem is \emph{NP}-complete~\cite{garey1979computers}.

\begin{proof}[Proof of Theorem~\ref{Thm:main.result.vers2}]
We first introduce the reduction $\mathcal{A}$ that maps problem instances of the partition problem to problem instances of $\AR$-Dim-$d$. Let $\vec{c} \in \mathbb{N}^Z$ be an input to the partition problem. Define $\vec{v}  \in \mathbb{N}^Z$ such that $(\vec{v})_j = (\vec{c})^2_j$ and set
\beq\label{fwefk34jrk}
	\mathcal{A}(\vec{c}) :=
	\left( \begin{array}{cccc}  I_Z & \vec{v}/\| \vec{v} \|_1 & \vec{1}/Z & *  \\  \vec{v}^T/\| \vec{v} \|_1 &  &  &   \\  \vec{1}^T/Z &   & I_Z &     \\  * &  &  &    \end{array} \right).
\eeq
More precisely, $\mathcal{A}(\vec{c})$ is the flattening~\eqref{def.data.D.0} of the input $\bigl( p_{xyz} \bigr)_{x,y,z \in \Omega}$ to $\AR$-Dim-$d$. The asterisks mark entries $\not\in \Omega$. To prove the Theorem~\ref{Thm:main.result.vers2} we need to show that $\vec{c}$ is a yes-instance for the partition problem if and only if $\mathcal{A}(\vec{c})$ is a yes-instance for $\AR$-Dim-$d$.

``$\Rightarrow$": Assume $\vec{c}$ is a \emph{yes}-instance for the partition problem. Let $\vec{s} \in \{ \pm 1 \}^Z$ be the valid sign configuration corresponding to $\vec{c}$, i.e., $\sum_j (\vec{s})_j (\vec{c})_j = 0$. Let $(\vec{e}_j)_{j = 1}^Z$ the canonical basis in $\mathbb{R}^Z$. For each $j \in [Z]$ define $\vec{\psi}_j = \vec{e}_j$, $\vec{\varphi}_2 = \vec{1}/\sqrt{Z}$ and for each $n \in [Z]$ set
\[	
	(\vec{\varphi}_1)_n := (\vec{s})_n (\vec{c})_n / \| \vec{c} \|_2.
\]
By construction,
\[
	\vec{\varphi}_1^T \vec{\varphi}_2 = \frac{1}{\| \vec{c} \|_2 \sqrt{Z}} \sum_{n=1}^Z (\vec{s})_n (\vec{c})_n = 0
\]
and $\| \vec{\varphi}_1 \|_2 = \| \vec{\varphi}_2 \|_2 = 1$. Therefore, $\{ \vec{\varphi}_1, \vec{\varphi}_2 \}$ can be completed to an orthonormal basis $\{ \vec{\varphi}_j \}_{j=1}^Z$ in $\mathbb{R}^Z$. For $x \in [Z]$ set $\rho_x = \vec{\psi}_x \vec{\psi}_x^T$ and for $x \in [2Z] \backslash [Z]$ define $\rho_x = \vec{\varphi}_{x-Z} \vec{\varphi}_{x-Z}^T$. For $z \in [Z]$ define $E_{1z} = \vec{\psi}_z \vec{\psi}_z^T$ and $E_{2z} = \vec{\varphi}_z \vec{\varphi}_z^T$. It follows that the states $(\rho_x)_{x=1}^{2Z}$ and the measurements $(E_{1z})_z$ and $(E_{2z})_z$ define a valid quantum model for $\AR$-Dim-$d$. We conclude that the reduction~\eqref{fwefk34jrk} maps \emph{yes}-instances for the partition problem to \emph{yes}-instances for $\AR$-Dim-$d$.

``$\Leftarrow$": Assume $\mathcal{A}(\vec{c})$ is a yes instance for $\AR$-Dim-$d$. In the remainder we use the double index $(1,x)$ to refer to $x$ if $x \in [Z]$ and we use the double index $(2,x-Z)$ to refer to $x$ if $x \in [2Z] \backslash [Z]$. For instance, $\rho_{Z+2} = \rho_{2,2}$. This leads to a symmetric labeling of the rows and the columns. By Cauchy-Schwarz
\beq\label{fej43jifdnn}
	\| \rho_{yz} \|_2 \| E_{yz} \|_2 \geq \tr(\rho_{yz} E_{yz}) = 1
\eeq
for all $y \in [2]$ and $z \in [Z]$. It follows that 
\beq\label{feij3iojjfe}
	\| E_{yz} \|_2 \geq 1
\eeq
because $\| \sigma \|_2 \in [1/\sqrt{Z},1]$ for all $Z$-dimensional quantum states $\sigma$. By $\sum_{z=1}^Z E_z = I_d$ and by self-duality of the cone of positive semidefinite matrices, 
	\beq\begin{split}\label{fkj456}
		Z 
		&= \| I \|^2_2 = \Bigl\| \sum_{z=1}^Z E_{yz} \|_2^2 \\
		&= \Bigl( \sum_{z=1}^Z \| E_{yz} \|_2^2 \Bigr) + \sum_{z \neq z'} \tr(E_{yz} E_{yz'})
		\geq \sum_{z=1}^Z \| E_{yz} \|_2^2 .
	\end{split}\eeq
	By~\eqref{fkj456}, there exists $z'$ such that $\| E_{yz'} \|_2 \leq 1$. Assume there exists $z^*$ such that $\| E_{yz^*} \|_2 > 1$. Then, by~\eqref{feij3iojjfe},
	\beq\begin{split}\nn
		Z 
		&\geq \sum_{z=1}^Z \| E_{yz} \|_2^2 \geq E_{yz^*} + (Z-1) \min_z \| E_{yz} \|^2_2 \\
		&>  Z \min_z \| E_{yz} \|^2_2 \geq Z.
	\end{split}\eeq
	This is impossible and therefore, $\| E_{yz} \|_2 = 1$ for all $y \in [2]$ and $z,z' \in [Z]$. By~\eqref{fej43jifdnn},
\beq
	\| \rho_{yz} \|_2 = 1,
\eeq
i.e., all states are pure and the Cauchy-Schwarz inequality~\eqref{fej43jifdnn} is satisfied with equality. This happens if and only if there exists $\kappa \in \AR$ such that $\rho_{yz} = \kappa E_{yz}$. From $\| \rho_{yz} \|_2 = \| E_{yz} \|_2$ we conclude that $\kappa \in \{ \pm 1 \}$. The possibility $\kappa = -1$ can be ruled out because both $\rho_{yz}$ and $E_{yz}$ are positive semidefinite. We conclude that for all $y \in [2]$ and $z \in [Z]$ there exist unit vectors $\vec{\psi}_{yz} \in \mathbb{R}^Z$ such that
\beq\label{lfklek5469kflrk}
	( \vec{\psi}_{yz}^T \vec{\psi}_{y'z'} )^2 = \mathcal{A}(\vec{c})_{yz,y'z'}.
\eeq
This implies that there exist signs $s_{yz,y'z'} \in \{ \pm 1 \}$ such that 
\[
	\rank\Bigl(  \bigl( s_{yz,y'z'} \sqrt{\mathcal{A}(\vec{c})_{yz,y'z'}}  \bigr)_{yz,y'z'} \Bigr) \leq Z
\]
and therefore, 
\beq\label{fdi345jl2k3}
	\rank\Bigl(  \bigl( s_{yz,y'z'} \sqrt{\mathcal{A}(\vec{c})_{yz,y'z'}}  \bigr)_{yz,y'z'} \Bigr) = Z
\eeq
because $\bigl( \sqrt{\mathcal{A}(\vec{c})_{yz,y'z'}} \bigr)_{yz,y'z'}$ contains $I_Z$ as submatrix. Moreover, by~\eqref{lfklek5469kflrk}, $\bigl( s_{yz,y'z'} \sqrt{\mathcal{A}(\vec{c})_{yz,y'z'}}  \bigr)_{yz,y'z'}$ is the Gram matrix of the vectors $(\vec{\psi}_{yz})_{yz}$ and thus, symmetric. Denote by $\mathcal{A}(\vec{c})'$ the submatrix of $\bigl( s_{yz,y'z'} \sqrt{\mathcal{A}(\vec{c})_{yz,y'z'}}  \bigr)_{yz,y'z'}$ formed by its first $Z+2$ rows and columns, i.e.,
\beq\label{flklklkekf}
	\mathcal{A}(\vec{c})' :=
	\left( \begin{array}{ccc}  I_d & \vec{x} & \vec{y}  \\  \vec{x}^T  &  1 & 0  \\  \vec{y} &  0 & 1 \end{array} \right),
\eeq
where $\vec{x}$ is a $Z$-dimensional vector satisfying
\[
	\bigl( \vec{x} \bigr)_j = s_{1,j;2,1} (\vec{c})_j / \| \vec{c} \|_2 = \vec{\psi}_{1,j}^T \vec{\psi}_{2,1}
\]
for all $j \in [Z]$ and $\vec{y}$ is a $Z$-dimensional vector satisfying
\[
	\bigl( \vec{y} \bigr)_j = s_{1,j;2,2}  / \sqrt{Z} = \vec{\psi}_{1,j}^T \vec{\psi}_{2,2}.
\]
for all $j \in [Z]$. Therefore,~\eqref{flklklkekf} implies that
\beq\begin{split}
	0 
	&= 	\vec{\psi}_{2,1}^T \vec{\psi}_{2,2} = \vec{\psi}_{2,1}^T I_Z \vec{\psi}_{2,2} = \sum_{j=1}^Z \bigl( \vec{\psi}_{2,1}^T \vec{\psi}_{1,j} \bigr) \bigl( \vec{\psi}_{1,j}^T  \vec{\psi}_{2,2} \bigr)\\
	&=	 \sum_{j=1}^Z (\vec{x})_j (\vec{y})_j = \frac{1}{\sqrt{Z} \; \| \vec{c} \|_2} \sum_{j=1}^Z s_{1,j;2,1} \; (\vec{c})_j \; s_{1,j;2,2}.
\end{split}\eeq
It follows that
\[
	0 = \sum_{j=1}^Z \hat{s}_{j} \; (\vec{c})_j 
\]
if we define
\[
	\hat{s}_j := s_{1,j;2,1} \; s_{1,j;2,2} \in \{ \pm 1 \}
\]
for all $j \in [Z]$. We conclude that $\vec{c}$ is a \emph{yes}-instance for the partition problem. This concludes the proof of the Theorem because we have shown that $\vec{c}$ is a \emph{yes}-instance for the partition problem if and only if $\mathcal{A}(\vec{c})$ is a yes instance for $\AR$-Dim-$d$.
\end{proof}

\section{Proof of Theorem~\ref{Thm:main.theorem.for.QM.bipartite.setting}}\label{sect:proof.of.bipartite.setting}

Recall the decision problems \emph{3col}, \emph{rank-3}$^{\Delta}$, \emph{Dim-3}$^{\text{(pre)}}$ and \emph{Dim-3$^{(AB)}$} from section~\ref{sect:proof.sketch}. Here, we use the proof strategy sketched in Figure~\ref{fig:direct.reduction.to.Dim3.via.3col} to prove Theorem~\ref{Thm:main.theorem.for.QM.bipartite.setting} in terms of a reduction $\mathcal{A}' = \mathcal{A}'_3 \circ \mathcal{A}_2 \circ \mathcal{A}_1$. Lemma~\ref{Lemma:gjerkgjrlk} and Lemma~\ref{Lemma:fiehfheferf} prove that $\mathcal{A}_2 \circ \mathcal{A}_1$ is a valid reduction. It is left to show that $\mathcal{A}'_3$ is a valid reduction from \emph{Dim-3}$^{\text{(pre)}}$ to \emph{Dim-3$^{(AB)}$}.

\begin{lemma}[Reduction~$\mathcal{A}_3'$]\label{lemma:fejfkjekfhkejhtkj}
	The following are equivalent.
	\begin{itemize}
	\item There exists $\psi_{yz} \in \mathbb{C}^3$ such that the matrix $\data$ defined by $\data_{yz;y'z'} := | \bar{\psi}_{yz}^T \psi_{y'z'} |^2$  fits $\Delta(G')$. Here, $y,y' \in |V'|$ and $z,z' \in [3]$.	
	\item For $Z=Z'=3$ and $Y = Y' = |V'|$ there exists a bipartite $3$-dimensional quantum model with the property that the matrix $M$ defined by $M_{yz;y'z'} := 3 \, \tr( \rho E_{yz} \otimes F_{y'z'} )$ fits $\Delta(G')$.
	\end{itemize}
\end{lemma}

\begin{proof} We prove the lemma for general dimensions $d$.	The statement of the lemma can be reproduced by setting $d=3$.

``$\Rightarrow$": We set $\rho = \ketbra{\Omega}$ with $\ket{\Omega} = \frac{1}{\sqrt{d}} \sum_{j=1}^d \ket{jj}$. Moreover, we define $E_{yz} = F_{yz} = \ketbra{\psi_{yz}}$ for all $y \in [|V'|]$ and $z \in [3]$. Then,
\beq\begin{split}\label{dwjk4dfeh5jkhd}
	\tr\bigl( \rho \, E_{yz} \otimes F_{y'z'} \bigr)
	&=	\frac{1}{d} | \braket{\psi_{yz}}{\psi_{y'z'}} |^2.
\end{split}\eeq
Hence, by assumption, $d \, \tr( \rho E^y_z \otimes F^{y'}_{z'} )$ fits $\Delta(G')$.

``$\Leftarrow$": Fix $y,y' \in [Y]$. By Lemma~\ref{lemma:lb.op.norm},	
\[
	1 = \frac{p_{y'z'y'z'}}{\sum_{j=1}^d p_{y'z'y'j}} \leq \| F_{y'z'} \|.
\]
for all $z' \in [d]$ and
\[
	1 = \frac{p_{yzyz}}{\sum_{j=1}^d p_{yjyz}} \leq \| E_{yz} \|.
\]
for all $z \in [d]$. By Lemma~\ref{lemma.fekrjf}, there exist orthonormal bases $(\eta^y_j)_j$ and $(\mu^{y'}_j)_j$ of $\mathbb{C}^d$ with the property that
\[
	\text{$E_{yz} = \ketbra{\eta^y_z}$ and $F_{y'z'} = \ketbra{\mu^{y'}_{z'}}$} 
\]
for all $z,z' \in [d]$. By assumption, the quantum model fits $\Delta(G')$. It follows that the matrix with entries
\beq\label{fweft4t4f}
	M_{yz;y'z'} = d \, \tr( \rho \ketbra{\eta^y_z} \otimes \ketbra{\mu^{y'}_{z'}})
\eeq
fits $\Delta(G')$. In particular,
\[
	\frac{\delta_{zz'}}{d} = \tr( \rho \ketbra{\eta^y_z} \otimes \ketbra{\mu^{y}_{z'}})
\]
for all $y \in [Y]$.

We can interpret $\rho$ as the Choi-Jamiolkowski state of a completely positive map $\mathcal{E}: \mathrm{Herm}(\mathbb{C}^d) \rightarrow \mathrm{Herm}(\mathbb{C}^d)$. Then, by Proposition~2.1 in~\cite{wolf2012quantum},
\beq\begin{split}\label{dwjk4h5jkhd}
	\frac{\delta_{zz'}}{d} 
	&=	\tr\Bigl[ \rho \ketbra{\eta^y_z} \otimes \ketbra{\mu^{y}_{z'}} \Bigr]\\
	&=	\frac{1}{d} \tr\Bigl[ \ketbra{\eta^y_z} \; \mathcal{E}\bigl( \ketbra{\mu^{y}_{z'}} \bigr) \Bigr] .
\end{split}\eeq
Thus,
\beq\label{fewfkvjkfer}
	\delta_{zz'} = \bra{\eta^y_z}  \mathcal{E}\bigl(  \ketbra{\mu^y_{z'}} \bigr) \ket{\eta^y_z} 
\eeq
and therefore,
\beq\label{fe43543fef}
	\mathcal{E}\bigl(  \ketbra{\mu^y_z} \bigr) = \ketbra{\eta^y_z}
\eeq
That is because of the following. Let $\lambda^{yz'}_j$ and $\ket{\lambda^{yz'}_j}$ be such that $\mathcal{E}\bigl(  \ketbra{\mu^y_{z'}} \bigr) = \sum_j \lambda^{yz'}_j \ketbra{\lambda^{yz'}_j}$. There necessarily exists $j^* \in [d]$ is such that $\lambda^{yz'}_{j^*} > 0$. Otherwise, $\mathcal{E}\bigl(  \ketbra{\mu^y_{z'}} \bigr) = 0$ in contradiction with~\eqref{fewfkvjkfer}. It follows that for all $z \in [d]$ with $z \neq z'$,
\beq\begin{split}\nn
	0 
	&= 	\bra{\eta^y_z}  \mathcal{E}\bigl(  \ketbra{\mu^y_{z'}} \bigr) \ket{\eta^y_z} 
	\geq	\bra{\eta^y_z}  \Bigl( \lambda^{yz'}_{j^*}  \ketbra{\lambda^{yz'}_{j^*}} \Bigr)  \ket{\eta^y_z}  \geq 0.
\end{split}\eeq
Therefore, $\braket{\lambda^{yz'}_{j^*}}{\eta^y_z} = 0$ for all $z \in [d]$ with $z \neq z'$. By orthonormality of the basis $(\ket{\eta^y_z})_j$ and normality of $\ket{\lambda^{yz}_{j^*}}$ this implies that $\ket{\lambda^{yz}_{j^*}} = e^{i \theta_{yz}} \ket{\eta^y_{z'}}$ where $e^{i \theta_{yz}}$ denotes a phase factor. In particular, $\mathcal{E}\bigl(  \ketbra{\mu^y_{z'}} \bigr)$ must be rank-1, i.e., there exists $\lambda^{yz'} > 0$ such that $\mathcal{E}\bigl(  \ketbra{\mu^y_{z'}} \bigr) = \lambda^{yz'} \ketbra{\eta^y_{z'}}$. Therefore,~\eqref{fe43543fef} is implied by~\eqref{fewfkvjkfer} with $z=z'$.

To conclude the proof we go back to~\eqref{fweft4t4f}. As in~\eqref{dwjk4h5jkhd},
\[
	M_{yz;y'z'} = \tr\Bigl[ \ketbra{\eta^y_z} \; \mathcal{E}\bigl( \ketbra{\mu^{y'}_{z'}} \bigr) \Bigr].
\]
By~\eqref{fe43543fef},
\beq\begin{split}\nn
	M_{yzy'z'}
	&=	\tr\Bigl[ \ketbra{\eta^y_z} \; \ketbra{\eta^{y'}_{z'}}\Bigr] 
	=	| \braket{\eta^y_z}{\eta^{y'}_{z'}}|^2 .
\end{split}\eeq
This suffices to prove the claim because $M$ fits $\Delta(G')$ by assumption.

\end{proof}

We expect the proof idea behind Lemma~\ref{lemma:fejfkjekfhkejhtkj} to be of interest in proving that problem~\eqref{fewk435hjfewf0} admits a reduction to problem~\eqref{fewk435hjssqfewf}.

\begin{corollary}\label{Corollary:fjkerwjJHfe}
	The following are equivalent:
	\begin{itemize}
	\item $G$ is 3-colorable.
	\item For $Z=Z'=3$ and $Y = Y' = |V'|$ there exists a bipartite $3$-dimensional quantum model with the property that the matrix $M$ defined by $M_{yz;y'z'} := 3 \, \tr( \rho E_{yz} \otimes F_{y'z'} )$ fits $\Delta(G')$.
	\end{itemize}
\end{corollary}

\begin{proof}
	The claim is the straightforward combination of the statements of Lemma~\ref{Lemma:gjerkgjrlk}, Lemma~\ref{Lemma:fiehfheferf} and Lemma~\ref{lemma:fejfkjekfhkejhtkj}.
\end{proof}

\begin{corollary}\label{Corollary:fewfjkjkjcc}
	Dim-3$^{(AB)}$ is NP-hard.
\end{corollary}

\begin{proof}
	Corollary~\ref{Corollary:fjkerwjJHfe} reduces \emph{3col} to Dim-3$^{(AB)}$. Therefore, Dim-3$^{(AB)}$ is \emph{NP}-hard because \emph{3col} is \emph{NP}-complete~\cite{garey1976some}.
\end{proof}

Corollary~\ref{Corollary:fewfjkjkjcc} is sufficient to prove \emph{NP}-hardness of \emph{MinDim}$^{(AB)}$ because \emph{Dim-3}$^{(AB)}$ can be reduced to \emph{MinDim}$^{(AB)}$ by checking whether or not the optimal dimension computed by \emph{MinDim}$^{(AB)}$ is $\leq 3$.

\begin{lemma}\label{lemma:lb.op.norm}
	$$\frac{p_{yzy'z'}}{\sum_{j=1}^{Z'} p_{yzy'j}} \leq \| F_{y'z'} \|$$ and analogously, $$\frac{p_{yzy'z'}}{\sum_{j=1}^{Z} p_{yjy'z'}} \leq \| E_{yz} \|$$
\end{lemma}

\begin{proof}

\beq\begin{split}\label{fenkjfkejfk}
	p_{yzy'z'} 
	&=	\tr(\rho E_{yz} \otimes F_{y'z'}) \\
	&=	\tr\bigl( \sqrt{E_{yz}} \otimes I \; \rho \; \sqrt{E_{yz}} \otimes I \; I \otimes F_{yz'}  \bigr) \\
	&= 	\tr\Bigl[  \tr_A\bigl( \sqrt{E_{yz}} \otimes I \; \rho \; \sqrt{E_{yz}} \otimes I \bigr)    F_{y'z'}     \Bigr]\\
	&= 	\tr\bigl[ \kappa_{yz} F_{y'z'} \bigr]
\end{split}\eeq
where
\[
	\kappa_{yz} := \tr_A\bigl( \sqrt{E_{yz}} \otimes I \; \rho \; \sqrt{E_{yz}} \otimes I \bigr).
\]
Note that
\beq\begin{split}\label{wefefekjjfkel24}
	\tr\bigl[ \kappa_{yz} \bigr]
	&=	\tr\Bigl[ \kappa_{yz} \sum_{z'} F_{y'z'} \Bigr]\\
	&=	\sum_{z'} \tr\Bigl[ \tr_A\bigl( \sqrt{E_{yz}} \otimes I \; \rho \; \sqrt{E_{yz}} \otimes I \bigr)  F_{y'z'} \Bigr]\\
	&=	\sum_{z'} \tr\Bigl[  \sqrt{E_{yz}} \otimes I \; \rho \; \sqrt{E_{yz}} \otimes I  \; I \otimes F_{y'z'} \Bigr]\\
	&=	\sum_{z'} \tr\Bigl[ \rho \; E_{yz} \otimes F_{y'z'} \Bigr]
	= 	\sum_{z'} p_{yzy'z'}.
\end{split}\eeq
By~\eqref{fenkjfkejfk} and~\eqref{wefefekjjfkel24},
\beq\begin{split}
	p_{yzy'z'}
	&=		\tr\Bigl[ \kappa_{yz} F_{y'z'} \Bigr]
	\leq		\tr( \kappa_{yz} ) \| F_{y'z'} \| \\
	&=		\Bigl( \sum_{j} p_{yzy'j} \Bigr) \| F_{y'z'} \| 
\end{split}\eeq
and therefore,
\[
	\frac{p_{yzy'z'}}{\sum_{j} p_{yzy'j}} \leq \| F_{y'z'} \| .
\]
\end{proof}

\begin{lemma}\label{lemma.fekrjf}
	Let $(E_j)_{j=1}^d$ denote a POVM on $\mathbb{C}^d$. Assume that for all $j \in [d]$, $\| E_j \| \geq 1$. Then, there exists $(\psi_j)_{j=1}^d$ orthonormal such that for all $j \in [d]$, $E_j = \ketbra{\psi_j}$.
\end{lemma}

\begin{proof}
	\beq\begin{split}\label{wefwwrefekjjfkel24}
	d 
	&=		\| I \|_2^2
	=		\Bigl\|  \sum_{j=1}^d E_j \Bigr\|_2^2
	=		\Bigl( \sum_{j=1}^d \| E_j \|_2^2 \Bigr) + \Bigl( \sum_{j \neq j'} \tr( E_j E_{j'} ) \Bigr)\\
	&\geq	\Bigl( \sum_{j=1}^d \| E_j \|_2^2 \Bigr)
	\geq	\Bigl( \sum_{j=1}^d \| E_j \|^2 \Bigr)
	\geq	d.
\end{split}\eeq
Thus,
\[
	 d - \Bigl( \sum_{j=1}^d \| E_j \|_2^2 \Bigr) = \sum_{j \neq j'} \tr( E_j E_{j'} ) \leq 0.
\]
The last inequality holds by assumption. This last inequality can only be satisfied with equality because $\tr(AB) \geq 0$ for all positive semidefinite matrices $A,B$.
\end{proof}

\section{Conclusions}\label{sect:conclusions}

We have shown that optimal quantum models cannot be computed efficiently from measured data. We proved this claim in both the natural 1-party (cf.~Theorem~\ref{Thm:main.theorem.for.complex.QM}) and the natural 2-party setting (cf.~Theorem~\ref{Thm:main.theorem.for.QM.bipartite.setting}). We proved \emph{NP}-hardness by reducing 3-coloring to the inference of quantum models.

What other questions remain in this field? In both Theorem~\ref{Thm:main.theorem.for.complex.QM} and Theorem~\ref{Thm:main.theorem.for.QM.bipartite.setting} we search for a quantum model which reproduces the measured probabilities exactly. Does the hardness result extend to situations where we are satisfied with only approximating the measured probabilities? And which classes of data $(p_{xyz})_{(xyz) \in \Omega}$ admit efficient inference? In regard of the latter question, it appears important to illuminate the \emph{tradeoff} between 
\begin{itemize}
\item the relevance of the class of considered datasets $\{ (p_{xyz})_{(xyz) \in \Omega} \}$ and
\item the computational hardness of inference associated to those datasets.
\end{itemize}
The hardness of the classical analog of \emph{MinDim} turns out to be much easier to prove as it directly reduces to the problem of computing the so called nonnegative rank which is known to be \emph{NP}-hard~\cite{vavasis2009complexity}.

\subsection*{Acknowledgement}

I would like to thank Aram Harrow, Robin Kothari, Lior Eldar, Pablo Parrilo, Jonathan Kelner, Ankur Moitra, L{\'i}dia del Rio, Adam Sawicki and Antonios Varvitsiotis for interesting discussions. In particular I would like to thank Matt Coudron for fruitful discussions and for pointing out reference~\cite{peeters1996orthogonal}. This work is supported by the ARO grant Contract Number W911NF-12-0486. This work is preprint MIT-CTP/4719.

\appendix
\section{Sketch of derivation from~\cite{peeters1996orthogonal}}\label{Sect:derivation.peeters}

For the reader's convenience, we provide a sketch of the arguments from~\cite{peeters1996orthogonal} that lead to the proof of Theorem~\ref{Thm:peeters}.

\begin{lemma}\label{fwefknewkfwejk}
	Every possible 3-coloring of $H_{ij}$ is either of the form 
	\begin{itemize}
	\item $\bigl\{ \{ i,c_{ij} \}, \{ a_{ij},d_{ij} \}, \{ b_{ij},j \} \bigr\}$,
	\end{itemize}
	or it is of the form 
	\begin{itemize}
	\item $\bigl\{ \{ i,j \}, \{ a_{ij},c_{ij} \}, \{ b_{ij},d_{ij} \} \bigr\}$.
	\end{itemize}
	Here, $\{ x,y \}$ means that $x,y \in H_{ij}$ share the same color.
\end{lemma}

\begin{proof}
	Inspect the graph $H_{ij}$.
\end{proof}	

\begin{corollary}\label{fwefkljefejejwkk3}
	$G$ is 3-colorable if and only if $G'$ is 3-colorable.
\end{corollary}

\begin{lemma}\label{fwfkk4kfke}
	Fix $i,j \in G$ and let $H_{ij}$ be as in Fig.~\ref{Fig:Peeters.gadget}. Choose an orthonormal basis in $\mathbb{C}^6$ and denote its elements by $e_i, e_{a_{ij}}, e_{b_{ij}}, e_{c_{ij}}, e_{d_{ij}},e_j$. With respect to that ordering of the basis elements, there are exactly two 3-parameter families of matrices which fit $H_{ij}$ and whose rank is 3, namely,
	\beq\begin{split}
	M &=
	\left( \begin{array}{cccccc}
  	1	&0	&0	&0	&0	&a   \\
	0	&1	&0	&b	&0	&0  \\
	0	&0	&1	&0	&c	&0   \\
	0	&b^{-1}	&0	&1	&0	&0   \\
	0	&0	&c^{-1}	&0	&1	&0   \\
	a^{-1}	&0	&0	&0	&0	&1  
\end{array} \right),\\
	M' &= \left( \begin{array}{cccccc}
  	1	&0	&0	&a	&0	&0   \\
	0	&1	&0	&0	&b	&0  \\
	0	&0	&1	&0	&0	&c   \\
	a^{-1}	&0	&0	&1	&0	&0   \\
	0	&b^{-1}	&0	&0	&1	&0   \\
	0	&0	&c^{-1}	&0	&0	&1  
\end{array} \right)
	\end{split}\eeq
	where $a,b,c \in \mathbb{C}\backslash \{0\}$.
\end{lemma}

\begin{proof}
	By definition~\ref{Def:fitting.of.a.graph} and Fig.~\ref{Fig:Peeters.gadget}, a matrix $A$ fitting $H_{ij}$ must be of the form
	\[
A = \left( \begin{array}{cccccc}
  	1	&0	&0	&a	&0	&b   \\
	0	&1	&0	&c	&d	&0  \\
	0	&0	&1	&0	&e	&f   \\
	u	&v	&0	&1	&0	&0   \\
	0	&w	&x	&0	&1	&0   \\
	y	&0	&z	&0	&0	&1  
\end{array} \right)
	\]
	for some complex scalars $a,b,c,d,e,f$ and $u,v,w,x,y,z$. Using Gaussian elimination, we arrive at
	\[
A = \left( \begin{array}{cccccc}
  	1	&0	&0	&a	&0	&b   \\
	0	&1	&0	&c	&d	&0  \\
	0	&0	&1	&0	&e	&f   \\
	0	&0	&0	&1-au-vc	&-dv	&-bu   \\
	0	&0	&0	&-cw	&1-dw-ex	&-fx   \\
	0	&0	&0	&-ay	&-ez	&1-by-fz  
\end{array} \right).
	\]
	Assume $\rank(A) \leq 3$. Then,
	\[
	A' :=\left( \begin{array}{ccc}
	1-au-vc	&-dv	&-bu   \\
	-cw	&1-dw-ex	&-fx   \\
	-ay	&-ez	&1-by-fz  
\end{array} \right) = 0.
	\]
	By $A'_{12}=0$, we have that $dv=0$. This leaves us with the following alternatives:
	\begin{itemize}
	\item		$d = 0$,
	\item		$v = 0$,
	\item		$d = 0$ and $v = 0$.
	\end{itemize}
	Consider the alternative ``$d = 0$". By $d = 0$ and $A'_{22} = 0$, 
	\beq\label{qfwe453f} 
		e \neq 0, \ x \neq 0.
	\eeq
	By $e \neq 0$ and $A'_{32} = 0$, $z = 0$. By $z = 0$ and $A'_{33} = 0$,
	\beq\label{ga435gdfg}
		b \neq 0, \ y \neq 0.
	\eeq
	By $b \neq 0$ and $A'_{13} = 0$,
	\beq\label{jtj54whdxxd}
		u=0.
	\eeq
	By $y \neq 0$ (cf.~\eqref{ga435gdfg}) and $A'_{31} = 0$, $a = 0$. By $x \neq 0$ (cf.~\eqref{qfwe453f}) and $A'_{23} = 0$, $f=0$. By $u =0$ (cf.~\eqref{jtj54whdxxd}) and $A'_{11} = 0$,
	\[
		v \neq 0, \ c \neq 0.
	\]
	By $c \neq 0$ and $A'_{21} = 0$, $w = 0$. Hence, there exist complex scalars $b,c,e,v,x,y$ such that
	\beq\label{fwknkekgrekg}
A = \left( \begin{array}{cccccc}
  	1	&0	&0	&0	&0	&b   \\
	0	&1	&0	&c	&0	&0  \\
	0	&0	&1	&0	&e	&0   \\
	0	&v	&0	&1	&0	&0   \\
	0	&0	&x	&0	&1	&0   \\
	y	&0	&0	&0	&0	&1  
\end{array} \right).
	\eeq
	Denote by $\vec{A}_j$ the $j$-th column of $A$. By~\eqref{fwknkekgrekg},
	\beq\begin{split}
		\mathrm{span}\{ \vec{A}_1,\vec{A}_6 \}  &\perp  \mathrm{span}\{ \vec{A}_2,\vec{A}_4 \}\\
		\mathrm{span}\{ \vec{A}_1,\vec{A}_6 \}  &\perp  \mathrm{span}\{ \vec{A}_3,\vec{A}_5 \}\\
		\mathrm{span}\{ \vec{A}_2,\vec{A}_4 \}  &\perp  \mathrm{span}\{ \vec{A}_3,\vec{A}_5 \}.				
	\end{split}\eeq
	Hence, if we violate any of the relations
	\beq\label{fwfkwfkle4535}
		\mathbb{C} \vec{A}_1 = \mathbb{C} \vec{A}_6, \ \mathbb{C} \vec{A}_2 = \mathbb{C} \vec{A}_4, \ \mathbb{C} \vec{A}_3 = \mathbb{C} \vec{A}_5
	\eeq
	we necessarily violate $\rank(A) = 3$. Therefore, by~\eqref{fwfkwfkle4535}, $\rank(A) = 3$ implies that for some complex and non vanishing scalars $a,b,c$, the matrix $A$ is of the form $M$ (see statement of the Lemma).  
	
	By a similar line of arguments we conclude that for the alternative ``$v = 0$", the matrix $A$ is of the form $M'$. Moreover, working out the details of the discussion of the alternative ``$v = 0$", it is easy to see that the remaining alternative ``$d = 0$ and $v = 0$" and the demand $A'_{11} = 0$ can not coexist. 
\end{proof}

\begin{theorem}[see~\cite{peeters1996orthogonal}]\label{Thm:peeters}
	$G$ is 3-colorable if and only if there exists a Gram matrix $A$ such that $\rank(A) \leq 3$ and such that $A$ fits $G'$.
\end{theorem}

\begin{proof}

``$\Rightarrow$": By Corollary~\ref{fwefkljefejejwkk3}, $G$ is 3-colorable if and only if $G'$ is 3-colorable. For all $j \in V'$, let $\mathrm{c}(j)$ denote the color of vertex $j \in V'$ as specified by a 3-coloring of $G'$. Define $P \in \AR^{3 \times |V'|}$ by
	\[
		P = ( \vec{e}_{\mathrm{c}(1)}, ..., \vec{e}_{\mathrm{c}(V')} ), 
	\]
	and set $A' := P^T P$. Then, $A'$ fits $G'$ and $\rank(A') = 3$. 

``$\Leftarrow$": By assumption there exists $d \in [3]$ and $P \in \mathbb{C}^{d \times | V' |}$ such that $A = \bar{P}^T P$ fits $G'$. For some $i,j \in V$, we denote by $A^{ij}_{_{\square}}$ the sub-Gram matrix of $A$ corresponding to the vertices $i,j,a_{ij}, b_{ij}, c_{ij}, d_{ij} \in V'$. Let $P^{ij}_{_{\square}}$ be the submatrix of \emph{P} with the property $A^{ij}_{_{\square}} = \bar{P}_{_{\square}}^{ij \, T} P^{ij}_{_{\square}}$. Denote by $W_i$, $W_{a_{ij}}$, $W_{b_{ij}}$, $W_{c_{ij}}$, $W_{d_{ij}}$, and $W_{j}$  the 1-dimensional subspaces spanned by the respective column vectors of $P_{_{\square}}$. 

Since $A^{ij}_{_{\square}}$ fits $H_{ij}$ we have by Lemma~\ref{fwfkk4kfke} that $A^{ij}_{_{\square}}$ is either of the form $M$ or $M'$ (as defined in Lemma~\ref{fwfkk4kfke}). Inspecting the alternatives $M$ and $M'$ we observe that we only encounter two possible scenarios, namely,
\begin{itemize}
\item \emph{Scenario `parallel'.} $W_i = W_j, W_{a_{ij}} = W_{c_{ij}}, W_{b_{ij}} = W_{d_{ij}}$. This happens if $A$ is of the form $M$.
\item \emph{Scenario `perpendicular'.} $W_i = W_{c_{ij}}, W_{a_{ij}} = W_{d_{ij}}, W_{b_{ij}} = W_j, \text{ and $W_i \perp W_j$}$. This happens if $A$ is of the form $M'$. 
\end{itemize}

Let $A_V$ be the sub-matrix of $A$ associated to the vertices $V \subseteq V'$. We have that $\rank(A_V) \leq \rank(A) \leq 3$. We treat the alternatives $\rank(A_V) = 1,2,3$ separately.

\emph{Case ``$\rank(A_V) = 1$".} By Lemma~\ref{fwfkk4kfke}, $A^{ij}_{_{\square}}$ is either of the form $M$ or $M'$. The possibility $M'$ is ruled out because it would lead to $\rank(A_V) \geq 2$. The only possibility is $A^{ij}_{_{\square}} = M$ with $a \neq 0$ (see definition of $M$ in Lemma~\ref{fwfkk4kfke}). Hence, all entries of $A_V$ are non-zero, i.e., the edge set of $G$ is empty (recall definition~\ref{Def:fitting.of.a.graph}). We conclude that that $G$ is 3-colorable.

\emph{Case ``$\rank(A_V) = 2$".} By $\rank(A_V) = 2$, there exist $i_1,i_2 \in V$ with $i_1 \neq i_2$ such that $W_{i_1} \neq W_{i_2}$. Therefore, `scenario parallel' is ruled out and `scenario perpendicular' applies, i.e., $W_{i_1} \perp W_{i_2}$. By $\rank(A_V) = 2$, $W_k \subset W_{i_1} \oplus W_{i_2}$ for all $k \in V$. Either `scenario parallel' or `scenario perpendicular' applies. Thus, either 
\[
	W_k = W_{i_1} \ \Rightarrow \ W_k \perp W_{i_2},
\]
or
\[
	W_k \perp W_{i_1} \ \Rightarrow \ W_k = W_{i_2}.
\]
On that basis, for all $k \in V$, we define the coloring of $G$ by
\beq\label{kewjfkejfknv}
	c(k) = 	\left\{ \begin{array}{ll} 
			\text{`r'},	& \text{if $W_k = W_{i_1}$}   \\ 
			\text{`g'},	& \text{if $W_k = W_{i_2}$.}
			\end{array} \right.
\eeq
We need to check that this is a valid 3-coloring. Assume $n,m \in V$ and $(n,m) \in E$. Then, by definition~\ref{Def:fitting.of.a.graph}, $A_{nm} = 0$. It follows that `scenario parallel' is ruled out, i.e., we have $W_n \perp W_m$. Hence, there cannot exist $l \in [2]$ such that $W_n = W_{i_s}$ and $W_m = W_{i_s}$. By~\eqref{kewjfkejfknv}, $c(n) \neq c(m)$.

\emph{Case ``$\rank(A_V) = 3$".} This case proceeds exactly as the previous case. Instead of two distinguished 1-dimensional subspaces we end up with three subspaces $W_{i_1}, W_{i_2}, W_{i_3}$ which are pairwise perpendicular. Again, for all $k \in V$, only one of the alternatives 
\beqa
	&W_k = W_{i_1}, \ W_k \perp W_{i_2}, \ W_k \perp W_{i_3}, \nn \\
	&W_k \perp W_{i_1}, \ W_k = W_{i_2}, \ W_k \perp W_{i_3}, \nn \\
	&W_k \perp W_{i_1}, \ W_k \perp W_{i_2}, \ W_k = W_{i_3} \nn
\eeqa
can apply. This motivates a coloring analogous to~\eqref{kewjfkejfknv}. The check that this is a valid 3-coloring is identical to the previous analysis of the ansatz~\eqref{kewjfkejfknv}.

\end{proof}



\end{document}